\documentclass[11pt]{article}
\usepackage{amsmath, amsthm, amssymb, psfrag}
\usepackage {mathrsfs, graphicx, newcent}

\newcommand{\be}{\begin{equation}}
\newcommand{\ee}{\end{equation}}
\newcommand{\bea}{\begin{eqnarray}}
\newcommand{\eea}{\end{eqnarray}}
\newcommand{\beas}{\begin{eqnarray*}}
\newcommand{\eeas}{\end{eqnarray*}}

\newcommand{\R}{\ensuremath{\mathbb R}}
\newcommand{\C}{\ensuremath{\mathbb C}}
\renewcommand{\P}{\ensuremath{\mathbb P}}
\newcommand{\Q}{\ensuremath{\mathbb Q}}

\newcommand{\E}{\ensuremath{\mathbb E}}
\newcommand{\cI}{\ensuremath{\mathcal{I}}}

\newcommand{\cJ}{\ensuremath{\mathcal{J}}}

\newcommand{\Ito}{It\ensuremath{\hat{\textrm{o}}}}
\newcommand{\ang}[1]{\ensuremath{ \left \langle #1 \right \rangle }}
\newcommand{\crl}[1]{\ensuremath{ \left\{ #1 \right\} }}
\newcommand{\edg}[1]{\ensuremath{ \left[ #1 \right] }}
\newcommand{\brak}[1]{\ensuremath{\left( #1 \right)}}

\renewcommand{\Re}{{\rm Re}}

\newcommand{\p}{\P}
\newcommand{\q}{\Q}
\newcommand{\Ex}[1]{\mathbb{E}_x \left[#1\right]}

\newtheorem{theorem}{Theorem}[section]

\newtheorem{corollary}[theorem]{Corollary}
\newtheorem{lemma}[theorem]{Lemma}
\newtheorem{remark}[theorem]{Remark}
\newtheorem{example}[theorem]{Example}
\newtheorem{examples}[theorem]{Examples}
\newtheorem{foo}[theorem]{Remarks}
\newenvironment{Example}{\begin{example}\rm}{\end{example}}

\newenvironment{Remark}{\begin{remark}\rm}{\end{remark}}

\title{Pricing and Hedging in Affine Models\\
with Possibility of Default\footnote{We thank Damir Filipovic, Ramon van Handel, Martin Keller-Ressel, 
Roger Lee and Ronnie Sircar for fruitful discussions and helpful comments.}}
\author{
Patrick Cheridito\footnote{Supported by NSF Grant DMS-0642361.}\\
Princeton University\\
Princeton, NJ 08544, USA
\and
Alexander Wugalter\footnote{Supported by NSF Grant DMS-0642361.}\\
Princeton University\\
Princeton, NJ 08544, USA
}
\date{First version: December 2010\\Current version: July 2011}

\begin{document}
\maketitle

\begin{abstract}
We propose a general framework for the simultaneous modeling of equity, government bonds, 
corporate bonds and derivatives. Uncertainty is generated by a general affine Markov process.
The setting allows for stochastic volatility, jumps, the possibility of
default and correlation between different assets. We show how to calculate
discounted complex moments by solving a coupled system of generalized Riccati equations.
This yields an efficient method to compute prices of power payoffs.
European calls and puts as well as binaries and asset-or-nothing options
can be priced with the fast Fourier transform methods of Carr and Madan (1999) and Lee (2005).
Other European payoffs can be approximated with a linear combination of government bonds, power
payoffs and vanilla options. We show the results to be superior to using
only government bonds and power payoffs or government bonds and vanilla options.
We also give conditions for European continent claims in our framework to be replicable
if enough financial instruments are liquidly tradable and study dynamic hedging strategies. 
As an example we discuss a Heston-type stochastic volatility model with possibility of default and 
stochastic interest rates.\\[2mm]
\end{abstract}

\setcounter{equation}{0}
\section{Introduction}
\label{sec:intro}

The goal of this paper is to provide a flexible class of models for the consistent pricing and
hedging of equity options, corporate bonds and government bonds. The noise in our models
is driven by an underlying Markov process that can generate stochastic volatility, jumps and default. For
the sake of tractability we assume it to be affine.
Then discounted complex moments of the underlying can be calculated
by solving a coupled system of generalized Riccati equations.
This yields an efficient method to compute prices of power payoffs and the
discounted characteristic function of the log stock price. From there,
prices of vanilla options as well as binaries and asset-or-nothing options can be
obtained with fast Fourier transform methods \`{a} la Carr and Madan \cite{CarrMadan} and Lee \cite{Lee}.
We also give conditions for European contingent claims in our models to be
replicable if there exist enough liquid securities that can be used as hedging instruments.

Our framework can be seen as an extension of the unified pricing and hedging model of
Carr and Schoutens \cite{CarrSchoutens}, where market completeness
is achieved through continuous trading in the money market, stock shares,
variance swaps and credit default swaps. The authors suggest to approximate general payoffs with
polynomials. In this paper we propose an approximations with a linear combination of government bonds,
non-integer power payoffs and European calls.

Affine models have become popular in the finance literature because they offer a
good trade-off between generality and tractability. One-factor affine processes were first used
by Vasicek \cite{Vasicek} and Cox--Ingersoll--Ross \cite{CIR} for interest rates modeling.
Popular affine stochastic volatility models include the ones by Stein and Stein \cite{SteinStein}
and Heston \cite{Heston}. For affine models in credit risk we refer to
Lando \cite{Lando}. Here we work with general affine processes in the sense of
Duffie et al. \cite{DFS}.

The rest of the paper is organized as follows. In Section 2 we introduce the model.
In Section 3 we show how discounted complex moments can be
calculated by solving generalized Riccati equations.
This yields an efficient way of calculating options with power payoffs.
As a corollary one obtains conditions for
the discounted stock price to be a martingale under the pricing measure.
The prices of vanilla options as well as
binaries and asset-or-nothing options can be computed with
the fast Fourier transform methods of Carr and Madan \cite{CarrMadan} and Lee \cite{Lee}.
For the pricing of European options with general payoffs we propose an
$L^2$-approximation with a linear combination of government bonds, power payoffs and European calls.
We illustrate this method by pricing a truncated log payoff, which can be applied
to the valuation of a variance swap in the case where the underlying can default. Section 4 is devoted
to the derivation of hedging rules. We show that in a model with
discrete jumps, or no jumps at all, every European option can
perfectly be hedged by trading in stock shares and a proper mix of European calls,
government and corporate bonds. A system of linear equations
is derived to find the hedging strategies. As an example we discuss in Section 5 a Heston-type
stochastic volatility model with the possibility of default and stochastic interest rates.

\setcounter{equation}{0}
\section{The model}
\label{sec:model}

Let $(X_t,\p_x)_{t \ge 0, \;x \in D}$ be a
time-homogeneous Markov process with values in $D := \mathbb{R}^m_+
\times \mathbb{R}^{n-m}$, $m \in \mathbb{N}$, $n \in \mathbb{N} \setminus \crl{0}$.
Let $\E_x$ be the expectation corresponding to $\p_x$ and denote by $\ang{\cdot,\cdot}$
the Euclidean (non-Hermitian) scalar product on $\C^n$:
$\left<x,y\right> : = \sum_{i=1}^n x_i y_i$, $x,y \in\C^n$. Furthermore, set
$\cI=\{1,\dots,m\}$ and $\mathcal{J}=\{m+1,\dots,n\}$. $X_{{\cal I},t}$ and $X_{{\cal J},t}$
denote the first $m$ and last $n-m$ components of $X_t$, respectively.

We model the risk-neutral evolution of the price of a stock share by
$$
S_t = \exp(s_t + R_t +\Lambda_t) 1_{\crl{t < \tau}},
$$
where
\begin{eqnarray*}
s_t &=& e + \ang{\varepsilon,X_t} \quad \mbox{for } (e,\varepsilon) \in \R\times\R^n\\
R_t &=& \int_0^t r_u du \quad
\mbox{for } r_t = d + \ang{\delta,X_{\cI,t}}, \quad
(d,\delta)\in\R_{+}\times\R^m_+\\
\Lambda_t &=& \int_0^t \lambda_u du \quad \mbox{for }
\lambda_t = c + \ang{\gamma,X_{\cI,t}}, \quad (c,\gamma) \in \R_+ \times \R^m_+
\end{eqnarray*}
and
$$
\tau = \inf \{t\ge0:\Lambda_t = E\}
$$
for a standard exponential random variable $E$ independent of $X$.

The process $r_t$ models the instantaneous risk-free interest rate and
$\lambda_t$ the default rate. Note that both of them are non-negative.
$\tau$ is the default time ($S_t = 0$ for $t \ge \tau$), and $s_t$ describes the
excess log-return of $S_t$ over $r_t + \lambda_t$ before default. 
Alternatively, one could model the three processes 
$\tilde{s}_t =  s_t + R_t + \Lambda_t$, $R_t$ and $\Lambda_t$.
But modeling $(s_t,R_t,\Lambda_t)$ is convenient since
it makes the discounted stock price equal to $\exp(s_t+\Lambda_t) 1_{\crl{t < \tau}}$
and will allow us to give a simple condition for it to be a martingale under $\p_x$
(see Corollary \ref{mp} below). 

For the sake of tractability we always assume that $(X_t,\P_x)_{t \ge 0,\;x \in D}$ is
stochastically continuous and affine in the sense that\\
(H1) $X_t \to X_{t_0}$ in probability for $t \to t_0$ with respect to all $\p_x$, $x \in D$, and\\
(H2) there exist functions
$\phi : \mathbb{R}_+ \times i \mathbb{R}^n \to \mathbb{C}$ and
$\psi : \mathbb{R}_+ \times i \mathbb{R}^n \to \mathbb{C}^n$ such that
$$
\E_x[\exp(\ang{u,X_t})] = \exp\left(\phi(t,u) + \ang{\psi(t,u),x}\right)
$$
for all $x \in D$ and $(t,u) \in \mathbb{R}_+ \times i\mathbb{R}^n$.

It has been shown by Keller-Ressel et al. \cite{KR2} that under
the hypotheses (H1) and (H2),
$(X_t,\P_x)_{t \ge 0,\;x \in D}$ automatically satisfies the regularity
condition of Definition 2.5 in Duffie et al. \cite{DFS}.
Therefore, according to Theorem 2.7 in Duffie et al. \cite{DFS}, it is a
Feller process whose infinitesimal generator ${\cal A}$ has
$C^2_c(D)$ (the set of twice continuously differentiable functions
$f: D \to \mathbb{R}$ with compact support) as a core and acts on
$f \in C^2_c(D)$ like
\bea
\notag
\mathcal{A}f(x) &=& \sum_{k,l=1}^n \left(a_{kl}+\left<\alpha_{\mathcal{I},kl},
x_\mathcal{I}\right>\right)\frac{\partial^2 f(x)}{\partial x_k\partial x_l}
+ \left<b+\beta x, \nabla f(x)\right>\\
\label{infgen}
&& +\int_{D \backslash\{0\}} \left(f(x+\xi)-f(x)-\left<\nabla_{\mathcal{J}} f(x),
\chi_{\mathcal J}(\xi)\right>\right) \nu(d\xi)\\
\notag
&& +\sum_{i=1}^m \int_{D\backslash\{0\}}
\left(f(x+\xi)-f(x)-\left<\nabla_{\mathcal{J}\cup\{i\}}f(x),\chi_{\mathcal{J}
\cup\{i\}}(\xi)\right>\right)x_i\mu_i(d\xi),
\eea
where $\chi = (\chi_1, \dots, \chi_n)$ for
$$
\chi_k(\xi)=
\begin{cases}
0 & \textrm{if $\xi_k=0$}\\
(1\wedge|\xi_k|)\frac{\xi_k}{|\xi_k|}, & \textrm{otherwise}
\end{cases}
$$
and
\begin{itemize}
\item[(i)] $a \in \R^{n \times n}$ is symmetric and positive
semi-definite with $a_{\mathcal{II}}=0$
\item[(ii)] $\alpha =(\alpha_i)_{i\in\mathcal I}$, where for each $i\in\mathcal I$,
$\alpha_i \in \R^{n\times n}$ is a positive
semi-definite symmetric matrix such that all entries in
$\alpha_{i,\mathcal{II}}$ are zero except for $\alpha_{i,ii}$
\item[(iii)] $b \in D$
\item[(iv)] $\beta \in \R^{n \times n}$ such that
$\beta_{\mathcal{IJ}}=0$ and $\beta_{\mathcal{II}}$ has
non-negative off-diagonal elements
\item[(v)] $\nu$ is a Borel measure on $D \backslash \{0\}$ with
$$\int_{D \backslash\{0\}} \left(\left<\chi_{\mathcal{I}}(\xi),
\mathbf{1}\right>+||\chi_{\mathcal{J}}(\xi)||^2\right) \nu(d\xi) < \infty$$
\item[(vi)] $\mu=(\mu_i)_{i\in\mathcal{I}}$ is a vector of Borel
measures on $D \backslash \{0\}$ such that
$$\int_{D\backslash\{0\}} \left(\left<\chi_{\mathcal{I}\backslash\{i\}}(\xi),
\mathbf{1}\right>+||\chi_{\mathcal{J}\cup\{i\}}(\xi)||^2\right) \mu_i(d\xi) < \infty.$$
\end{itemize}
In addition we require the jumps to satisfy the following exponential
integrability condition: For all $q \in \R^n$,
\begin{itemize}
\item[(vii)]
$$
\int_{D} e^{\left<q,\xi\right>}1_{\{||\xi||\ge1\}}\nu(d \xi)<\infty
\quad \mbox{and} \quad
\int_{D} e^{\left<q,\xi\right>}1_{\{||\xi||\ge1\}}\mu_i(d\xi)<\infty
\quad \mbox{for all $i\in\cI$}.
$$
\end{itemize}
It follows from Theorem 2.7 and Lemma 9.2 in Duffie et al. \cite{DFS}
that for every set of parameters $(a,\alpha,b,\beta,\nu,\mu)$
satisfying (i)--(vii), \eqref{infgen} defines the infinitesimal generator of a $D$-valued
Feller process $(X_t,\P_x)_{t \ge 0,\;x \in D}$ satisfying conditions (H1) and (H2)
for $(\phi, \psi)$ equal to the unique solution of the following system of generalized Riccati equations
\begin{align}
\label{RicPhi}
\begin{cases}
\partial_t \phi(t,u)=F_0(\psi(t,u)), \quad \phi(0,u) = 0\\
\partial_t \psi_{\cI}(t,u)= F(\psi(t,u)), \quad \psi_{\cI}(0,u) = u_{\cI}\\
\psi_{\cJ}(t,u)=\exp(\beta^T_{\mathcal{JJ}}t)u_\cJ
\end{cases},
\end{align}
where the functions $F_0 : \mathbb{C}^n \to \mathbb{C}$ and
$F : \mathbb{C}^n \to \mathbb{C}^m$ are given by
\beas
F_0(u)&=&\left<au,u\right> + \left<b,u\right>
+ \int_{D\backslash\{0\}}\left(e^{\left<u,\xi\right>}-1
-\left<u_{\cJ},\chi_\mathcal{J}(\xi)\right>\right)\nu(d\xi)\\
F_i(u)&=& \left<\alpha_i u,u\right> + \sum_{k =1}^n \beta_{k i} u_k
+\int_{D \backslash\{0\}}\left(e^{\left<u,\xi\right>}-1
-\left<u_{\mathcal{J}\cup\{i\}},\chi_{\mathcal{J}\cup\{i\}}(\xi)\right>\right)\mu_i(d\xi),
\quad i \in \cI.
\eeas
In particular, $(a,\alpha,b,\beta,\nu,\mu)$ uniquely determine the transition
probabilities of $X$ and therefore, its distribution under each $\p_x$, $x \in D$.
Since every Feller process has a RCLL version, we may assume
$(X_t, \tau, \p_x)_{t \ge 0, x \in D}$ to be defined on
${\cal R}_D \times \mathbb{R}_+$, where ${\cal R}_D$ is the space of all
RCLL functions $\omega : \mathbb{R}_+ \to D$, and
$(X_t, \tau)(\omega,y) = (\omega(t),y)$.

\setcounter{equation}{0}
\section{Pricing}
\label{sec:pricing}

We price derivatives on $S_t$ by taking expectation under $\P_x$.
For instance, we determine the price of a European option with payoff function
$\varphi: \R_+ \rightarrow \R$ and maturity $t>0$ by
$$
\Ex{\exp\left(-R_t\right)\varphi(S_t)}.
$$
Special cases include:
\begin{itemize}
\item Government bonds: $\varphi \equiv 1$
\item Corporate bonds: $\varphi(x)= 1_{\crl{x > 0}}$
\item Call options: $\varphi(x)=(x-K)^+$ for $K>0$
\item Power payoffs: $\varphi(x)= x^p 1_{\{x > 0\}}$ for $p \in \R$.
\end{itemize}
General power payoffs are not traded. But they can be priced efficiently
and are helpful in pricing other payoffs. Moreover, for $p=1$, the price of the power payoff
equals the stock price and having a general formula for $\Ex{\exp(-R_t)S_t}$
will allow us to obtain conditions for the discounted stock price
to be a martingale. The case $p=0$ corresponds to a corporate bond. We
just consider zero coupon bonds and assume there is no recovery in the case of default.
A corporate bond with a fixed recovery can be seen as a linear combination of a
government and a corporate bond with no recovery. So all of the arguments that
follow can easily be extended to this case.

In Subsection \ref{sub:disc} we show how discounted moments of $S_t$
can be obtained by solving coupled systems of generalized Riccati equations.
In Subsection \ref{sub:fft} we extend Fourier pricing methods from Carr and Madan \cite{CarrMadan}
and Lee \cite{Lee} to our setup. In Subsection \ref{sub:approx} we discuss the
approximation of general European options with
government bonds, power payoffs and European calls.

\subsection{Discounted moments and generalized Riccati equations }
\label{sub:disc}

For $t \ge 0$ and $x\in D$, define
$$
U_{t,x} := \left\{z \in \C : \Ex{\exp(-R_t) S_t^{\Re(z)} 1_{\{t<\tau\}}} <\infty\right\}
$$
and the discounted moments function
$$
h_{t,x}(z)
:= \Ex{\exp(- R_t) S_t^z 1_{\crl{t<\tau}}}, \quad z \in U_{t,x}.
$$
Note that $U_{t,x}$ is equal to $\mathbb{C}$, a half-plane or a vertical strip. Since $R_t \ge 0$,
it always contains the imaginary axis $i \mathbb{R}$. Clearly, $h_{t,x}$ is finite on $U_{t,x}$ and
analytic in the interior of $U_{t,x}$. Indeed, $h_{t,x}(z)$ can be written as
$$
h_{t,x}(z) = \Ex{\exp \brak{z s_t + (z-1) R_t + z \Lambda_t} 1_{\crl{t < \tau}}}.
$$
So for $z$ in the interior of $U_{t,x}$, one can differentiate inside of the
expectation to obtain
$$
h^\prime_{t,x}(z) = \Ex{(s_t + R_t + \Lambda_t) \exp \brak{z s_t + (z-1) R_t
+ z \Lambda_t} 1_{\crl{t < \tau}}}.
$$
Let
$$
A : \mathbb{R}_+ \times \C^n \times \C \times \C \to \C \cup \crl{\infty}, \quad
B : \mathbb{R}_+ \times \C^n \times \C \times \C \to (\C \cup \crl{\infty})^n
$$
be solutions to the system of generalized Riccati equations
\begin{align}
\label{RicA}
\begin{cases}
\partial_t A(t,u,v,w)= G_0(B(t,u,v,w),v,w), \quad A(0,u,v,w)=0,\\
\partial_t B_\cI(t,u,v,w)= G(B(t,u,v,w),v,w), \quad B_\cI(0,u,v,w)=u_\cI,\\
B_\cJ(t,u,v,w)=\exp(\beta^T_{\mathcal{JJ}}t)u_\cJ
\end{cases}
\end{align}
where the functions $G_0 : \mathbb{C}^n \times \mathbb{C} \times \mathbb{C} \to \mathbb{C}$ and $G :
\mathbb{C}^n \times \mathbb{C} \times \mathbb{C} \to \mathbb{C}^m$ are given by
\beas
G_0(u,v,w)&=&\left<au,u\right> + \left<b,u\right>+ dv +c(w-1)
+ \int_{D\backslash\{0\}}\left(e^{\left<u,\xi\right>}-1
-\left<u_{\cJ},\chi_\mathcal{J}(\xi)\right>\right)\nu(d\xi)\\
G_i(u,v,w)&=& \left<\alpha_iu,u\right> + \sum_{k =1}^n \beta_{k i}
u_k +\delta_iv+\gamma_i(w-1)\\
&& +\int_{D \backslash\{0\}}\left(e^{\left<u,\xi\right>}-1
-\left<u_{\mathcal{J}\cup\{i\}},\chi_{\mathcal{J}\cup\{i\}}(\xi)\right>\right)\mu_i(d\xi),
\quad i \in \cI.
\eeas
The quadratic and exponential terms in $G_i$ can cause the solution
$(A,B)$ of \eqref{RicA} to explode in finite time. However, since
in condition (vii) we assumed the jumps of the
process $X$ to have all exponential moments, we obtain from Lemma
5.3 of Duffie et al. \cite{DFS} that the functions
$G_i$, $i = 0, \dots, m$, are analytic on $\C^{n+2}$.
In particular, they are locally Lipschitz-continuous, and it follows from the
Picard--Lindel\"of Theorem that for every $(u,v,w) \in \C^n \times \C \times \C$ there exists
$t^* >0$ such that equation \eqref{RicA} has a unique finite solution $(A,B)$ for $t \in [0,t^*)$.
We set $A$ and all components of $B$ equal
to $\infty$ after the first explosion time.
Note that the components $B_j(t,z\varepsilon,z-1,z)$, $j \in \cJ$, are finite and analytic
in $z \in \mathbb{C}$ for all $t \ge 0$, and
if $B_i(t,z\varepsilon,z-1,z)$ is finite/analytic in $z$ for all
$i \in \cI$, then so is $A(t,z\varepsilon,z-1,z)$. So for fixed $t \ge 0$, we define
$$
V_t := \{z \in \C: B_i(t,z\varepsilon,z-1,z) \mbox{ is finite for all } i \in \cI \}.
$$
and
\be \label{defl}
l_{t,x}(z) := \exp \crl{ze+ A(t,z\varepsilon,z-1,z)+\left<B(t,z\varepsilon,z-1,z),x\right>}, \quad z \in V_t.
\ee
In a first step we show the following

\begin{lemma} \label{lemma1}
For every $t \ge 0$, $V_t$ contains $i \mathbb{R}$ and
$h_{t,x}(z) = l_{t,x}(z)$ for all $x \in D$ and $z \in i \mathbb{R}$.
\end{lemma}

\begin{proof}
Fix $\p_x$, $x_{n+1}, x_{n+2} \in \mathbb{R}_+$ and denote
$\hat{x} = (x,x_{n+1},x_{n+2})$. Consider the process
$$
Y_t = \left\{
\begin{array}{cl}
(X_t, x_1 +R_t,x_2 +\Lambda_t) & \textrm{for }t<\tau\\
\Delta & \textrm{for } t \ge \tau
\end{array} \right.,$$
where $\Delta$ is a cemetery state outside of
$\hat{D} := \mathbb{R}_+^m \times \mathbb{R}^{n-m} \times \mathbb{R}^2_+$
to which $Y$ jumps at time $\tau$. $Y$ is a Markov process with values in $\hat{D} \cup \crl{\Delta}$.
Since $X$ has the properties (H1)--(H2), $Y$ fulfills the assumptions of
Proposition III.2.4 of Revuz and Yor \cite{RY} and therefore is a Feller process. Moreover,
$R_t = \int_0^t r_u du$ and $\Lambda_t = \int_0^t \lambda_u du$ are of finite variation
and the random variable $E$ is independent of $X$. So one deduces from (H1) that for all
$\hat{x} \in \hat{D}$ and $f \in C^2_c(\hat{D})$,
\beas
&& \frac{1}{t} \E_x[f(X_t,x_{n+1}+R_t,x_{n+2}+\Lambda_t) 1_{\crl{t < \tau}} - f(x,x_{n+1},x_{n+2})]\\
&=& \frac{1}{t} \E_x \edg{f(X_t,x_{n+1},x_{n+2}) - f(x,x_{n+1},x_{n+2})}\\
&& + \frac{1}{t} \E_x \edg{f(X_t,x_{n+1}+R_t,x_{n+2}+\Lambda_t) e^{-\Lambda_t} - f(X_t,x_{n+1},x_{n+2})}\\
&\to& \hat{\cal A}f(\hat{x}) \quad \mbox{for } t \downarrow 0,
\eeas
where
\beas
\hat{\mathcal{A}} f(\hat{x}) &=& \sum_{k,l=1}^n \left(a_{kl}
+\left<\alpha_{{\mathcal{I}},kl},x_{{\mathcal{I}}}\right>\right)
\frac{\partial^2 f(\hat{x})}{\partial x_k\partial x_l}
+ \left<\hat{b}+ \hat{\beta} \hat{x}, \nabla f(\hat{x})\right>
-(c+\left<\gamma,x_{{\cI}}\right>)f(\hat{x})\\
&& +\int_{D \setminus \{0\}} \left(f(\hat{x}+(\xi,0,0))-f(\hat{x})-\left<\nabla_{\mathcal{J}} f(\hat{x}),
\chi_{\mathcal J}(\xi)\right>\right) \nu(d\xi)\\
&& +\sum_{i=1}^m \int_{D \backslash\{0\}}
\left(f(\hat{x}+ (\xi,0,0))-f(\hat{x})-\left<\nabla_{\mathcal{J} \cup\{i\}} f(\hat{x}),
\chi_{\mathcal{J} \cup\{i\}} (\xi)\right>\right)x_i \mu_i (d \xi),
\eeas
and
\begin{itemize}
\item $\hat{b} = (b,d,c)$
\item $\hat{\beta} =
\left(
\begin{array}{cccc}
\beta_{\cI \cI} & 0 & 0 & 0\\
\beta_{\cJ \cI} & \beta_{\cJ\cJ} & 0 & 0\\
\delta & 0 & 0 & 0\\
\gamma & 0 & 0 & 0
\end{array}\right)\in\R^{(n+2)\times(n+2)}$.
\end{itemize}
By Lemma 31.7 of Sato \cite{Sato}, the infinitesimal generator of $Y$
is well-defined and equal to $\hat{\cal A}$ on $C^2_c(\hat{D})$
(if $f(\Delta)$ is understood to be $0$ for $f \in C^2_c(\hat{D})$).
But since $\hat{\cal A}$ is the infinitesimal generator of a regular affine process
with values in $\hat{D} \cup \crl{\Delta}$, it follows from Theorem 2.7 of Duffie et al. \cite{DFS}
that
\be \label{fDFS}
\E_x \left[\exp\left(\left<u,X_t\right>+vR_t +w\Lambda_t\right) 1_{\{t < \tau\}}\right]
= \exp\left(A(t,u,v,w)+\left<B(t,u,v,w),x\right>\right).
\ee
for all $(t,u,v,w) \in \mathbb{R}_+ \times \mathbb{C}_-^m \times (i \mathbb{R})^n \times
\mathbb{C}_-^2$, where $\mathbb{C}_-$ denotes the set of all $z \in \mathbb{C}$ with $\mbox{Re}(z) \le 0$.
In particular, for $(t,u,v,w) \in \mathbb{R}_+ \times \mathbb{C}_-^m \times (i \mathbb{R})^n \times
\mathbb{C}_-^2$, $A(t,u,v,w)$ and all components of $B(t,u,v,w)$ are finite.
In fact, according to Theorem 2.7 of Duffie et al. \cite{DFS}, the function
$B$ should have $n+2$ components and $G$ should have $m+2$. But due to the
special form of $\hat{\cal A}$,
components $m+1$ and $m+2$ of $G$ vanish, and the corresponding components of $B$ stay equal to the
initial values $v$ and $w$ for all times $t \ge 0$. In addition, $R_0 = \Lambda_0 = 0$, and one
obtains \eqref{fDFS}, where the functions $A$ and $B$ solve \eqref{RicA}.

Now notice that
$$
h_{t,x}(z)= \exp(ze)
\E_x \left[\exp\left(z \left<\varepsilon,X_t\right>+ (z-1)R_t +z\Lambda_t\right) 1_{\{t < \tau\}}\right].
$$
So, for all $t \ge 0$, $V_t$ contains $i \mathbb{R}$ and
$$
h_{t,x}(z) = \exp \crl{ze+ A(t,z\varepsilon,z-1,z)+\left<B(t,z\varepsilon,z-1,z),x\right>}
$$
for all $x \in D$ and $z \in i \mathbb{R}$.
\end{proof}

\begin{lemma} \label{lemma2}
Fix $t_0 \ge 0$ and $z_0 \in V_{t_0}$. Then there exists an open neighborhood
$W$ of $(t_0,z_0)$ in $\mathbb{R}_+ \times \mathbb{C}$ such that
$$
A(t,z \varepsilon,z-1,z) \quad \mbox{and} \quad B(t,z\varepsilon,z-1,z)
$$
are finite for $(t,z) \in W$ as well as analytic in $t$ and $z$. In particular, the set
$$
\crl{(t,z) \in \mathbb{R}_+ \times \mathbb{C} : z \in V_t}
$$
is open and for all $x \in D$, $l_{t,x}(z)$ is analytic in $t$ and $z$.
\end{lemma}

\begin{proof}
By Lemma 5.3 of Duffie et al. \cite{DFS}, the functions $G_i$, $i = 0, \dots, m$, are analytic on $\C^{n+2}$.
So, the lemma is a consequence of Theorem 1.1 in Ilyashenko and Yakovenko \cite{IY}.
\end{proof}

In the following theorem we extend the identity $h_{t,x}(z) = l_{t,x}(z)$
to $z$ outside of the imaginary axis $i \mathbb{R}$. Similar results
in different settings have been given by Filipovic and Mayerhofer \cite{FilipovicMayerhofer},
Keller-Ressel \cite{KR}, Spreij and Veerman \cite{SpreijVeerman}.

Denote by $I_t$ the largest interval around $0$ contained in
$V_t \cap \mathbb{R}$ and by $V^0_t$ the connected component of $V_t$ containing $0$.
It follows from Lemma \ref{lemma2} that $I_t$ is open in $\mathbb{R}$, and it is clear that
$I_t \subset V^0_t$. Moreover, one has

\begin{theorem} \label{thmmain}
For all $(t,x) \in \mathbb{R}_+ \times D$, $U_{t,x}$ is an open subset of $\mathbb{C}$ containing
the strip $\crl{z \in \mathbb{C} : \Re(z) \in I_t}$ and $h_{t,x}(z) = l_{t,x}(z)$
for each $z \in U_{t,x} \cap V^0_t$.
\end{theorem}

\begin{proof}
For fixed $t \ge 0$ and $x \in D$, one can write
$$
h_{t,x}(iy) = \Ex{\exp(-R_t) 1_{\crl{t < \tau}}}
\mathbb{E}_{\q_x} \edg{\exp(iy [s_t + R_t + \Lambda_t])}, \quad y \in \mathbb{R},
$$
where $\q_x$ is the probability measure given by
$$
\frac{d \q_x}{d \p_x} = \frac{\exp(-R_t) 1_{\crl{t < \tau}}}{
\Ex{\exp(-R_t) 1_{\crl{t < \tau}}}}.
$$
So, up to a constant, $y \mapsto h_{t,x}(i y)$ is the characteristic function of
$s_t + R_t + \Lambda_t$ with respect to $\q_x$. From Lemmas \ref{lemma1} and \ref{lemma2} we know that
$h_{t,x}(iy) = l_{t,x}(iy)$ for all $y \in \mathbb{R}$ and $V^0_t$ is an open
neighborhood of $0$ in $\C$ on which $l_{t,x}$ is analytic. So it follows from
Theorem 7.1.1 of Lukacs \cite{Lukacs} that $U_{t,x}$ contains the strip
$\crl{z \in \mathbb{C} : \Re(z) \in I_t}$ and is open. Since
$h_{t,x}(iy) = l_{t,x}(iy)$ for all $y \in \mathbb{R}$, $V^0_t$ is connected
and both functions are analytic, one also has
$h_{t,x}(z) = l_{t,x}(z)$ for $z \in U_{t,x} \cap V^0_t$.
\end{proof}

\begin{Remark}
The price of a corporate zero coupon bond with no recovery in the case of default is given by
$$
P_{t,x}= \Ex{e^{-R_t}1_{\{S_t>0\}}} = h_{t,x}(0).$$
The price of a government bond is equal to the price of
a corporate bond in a model where $(c,\gamma) = 0$ (that is, the probability of
default is zero).
\end{Remark}

\begin{corollary} \label{mp}
The condition
\be \label{martcond}
G_i(\varepsilon,0,1)=0, \quad i =0, \dots, m, \quad \mbox{and} \quad \beta_{\cJ\cJ}=0,
\ee
is sufficient for the discounted stock price $\exp(s_t + \Lambda_t) 1_{\crl{t < \tau}}$
to be a martingale with respect to all $\P_x$, $x \in D$. If all components of
$\varepsilon_{\cJ}$ are different from $0$, then \eqref{martcond} is also necessary.
\end{corollary}

\begin{proof}
It follows from Lemma \ref{lemma2} that there exists a $t_0 > 0$ such that
$p \in V^0_t$ for all $(t,p) \in [0,t_0] \times [0,1]$. So if \eqref{martcond} holds,
one obtains from Theorem \ref{thmmain} and \eqref{RicA} that
$$
h_{t,x}(1) = \exp(e + A(t,\varepsilon,0,1)+\left<B(t,\varepsilon,0,1),x\right>)
= h_{0,x}(1) \quad \mbox{for all } (t,x) \in [0,t_0] \times D.
$$
Now the martingale property of $\exp(s_t + \Lambda_t) 1_{\crl{t < \tau}}$
with respect to all $\P_x$, $x \in D$, follows by decomposing a given interval $[0,t]$
into finitely many intervals of length smaller than $t_0$ and taking iterative conditional expectations.
If all components of $\varepsilon_{\cJ}$ are
different from $0$ and condition \eqref{martcond} is violated, there exist $t \le t_0$ and
$x \in D$ such that $h_{t,x}(1) \neq h_{0,x}(1)$. So the martingale property cannot
hold under $\p_x$.
\end{proof}

\subsection{Pricing via Fourier transform}
\label{sub:fft}

In this subsection we show how to extend Fourier pricing methods from
Carr and Madan \cite{CarrMadan} and Lee \cite{Lee} to our setting.
Since $\p_x$ is used as pricing measure, one would typically work with
models in which the discounted stock price $\exp(s_t + \Lambda_t) 1_{\crl{t < \tau}}$
is a martingale under $\p_x$. However, the following results just involve
$S_t$ for fixed $t \ge 0$ and technically
do not need the discounted stock price to be a martingale.

Consider a call option with log strike $k$ and price
$$
c_{t,x}(k) = \Ex{e^{-R_t} \brak{S_t - e^k}^+}.
$$

\begin{theorem}
\label{thmFFT}
Let $(t,x) \in \mathbb{R}_+ \times D$ and $p > 0$ such that $p+1 \in U_{t,x}$. Then
\be \label{fft}
c_{t,x}(k) = \frac{e^{-p k}}{2 \pi}
\int_\R e^{-iyk} g_c(y) dy = \frac{e^{-p k}}{\pi}
\int_0^{\infty} \Re\left(e^{-iyk} g_c(y)\right) dy,
\ee
where
$$
g_c(y) = \frac{h_{t,x}(p+1+iy)}{p^2 + p - y^2 + iy(2p+1)}.
$$
\end{theorem}

\begin{proof}
Since $U_{t,x}$ is open, there exists $\eta >0$ such that
$p+1 + \eta \in U_{t,x}$, and it follows from Corollary 2.2
in Lee \cite{MomentFormula} that
$$c_{t,x}(k)=O(e^{-(p+\eta)k}) \quad \mbox{for } k \to \infty.$$
In particular,
$$
\int_{\R} \brak{e^{p k} c_{t,x}(k)}^2 dk < \infty.
$$
It follows that the Fourier transform
$$
g_c(y) = \int_{\mathbb{R}} e^{iyk} e^{p k} c_{t,x}(k) dk
$$
is square-integrable in $y$, and one can transform back to obtain
$$
e^{p k} c_{t,x}(k) = \frac{1}{2 \pi} \int_{\mathbb{R}} e^{- iyk} g_c(y) dy.
$$
This shows the first equality of \eqref{fft}. Since
$c_{t,x}(k)$ is real-valued, one has $g_c(-y) = \overline{g_c(y)}$,
which implies the second inequality of \eqref{fft}. To conclude the proof, set
$z = p +iy$ and note that
\beas
g_c(y) &=& \int_{\mathbb{R}} e^{zk} \Ex{ \brak{e^{s_t + \Lambda_t} - e^{k-R_t}}
1_{\crl{s_t + R_t + \Lambda_t \ge k,\; t< \tau}}} dk\\
&=& \Ex{\int_{-\infty}^{s_t + R_t + \Lambda_t}
\brak{e^{zk + s_t + \Lambda_t} - e^{(z+1) k -R_t}} dk 1_{\crl{t< \tau}}}\\
&=& \frac{1}{z(z+1)} \Ex{\exp(-R_t+(z+1)\log S_t) 1_{\crl{t< \tau}}}\\
&=& \frac{1}{z(z+1)} h_{t,x}(z+1).
\eeas
\end{proof}

To calculate prices of options with short maturities or extreme
out-of-the money strikes Carr and Madan \cite{CarrMadan} suggest an alternative
method which does not suffer from high oscillations. To adapt it to our setup, define
$$
d_{t,x}(k) :=
\begin{cases}
\Ex{e^{-R_t} \brak{e^k - S_t}^+1_{\{t<\tau\}}} & \textrm{if }k<\log S_0\\ \\
\Ex{e^{-R_t} \brak{S_t - e^k}^+} & \textrm{if }k>\log S_0
\end{cases}.
$$
Then the following holds:

\begin{theorem}
Let $(t,x) \in \mathbb{R}_+ \times D$ and $p>0$ such that $1-p,\;1+p \in U_{t,x}$. Then
\begin{align}
\label{CM2}
d_{t,x}(k) = \frac{1}{\sinh(p k)} \frac{1}{2 \pi} \int_{\mathbb{R}} e^{-iyk} g_d(y)dy,
\end{align}
where
$$
g_d(y)= \frac{f(y-ip)-f(y+ip)}{2}
$$
and
$$
f(y) = \frac{\exp((1+iy)s_0)}{1+iy}h_{t,x}(0)-\frac{\exp(iys_0)}{iy}h_{t,x}(1)
-\frac{h_{t,x}(1+iy)}{y^2-iy}.
$$
\end{theorem}

\begin{proof}
Since $U_{t,x}$ is open, there exists $\eta > 0$ such that
$1 - p - \eta$ and $1+ p + \eta$ belong to $U_{t,x}$.
By Corollary 2.2 of Lee \cite{MomentFormula},
$$
d_{t,x}(k)=O(e^{-(p+\eta) |k|}) \quad \mbox{for } k \to \pm \infty.
$$
In particular, $d_{t,x}(k)$ and $\sinh(pk) d_{t,x}(k)$ are both
square-integrable in $k$. One easily checks that
$$
\int_{\mathbb{R}} e^{iyk} d_{t,x}(k) dk = f(y).
$$
So,
$$
\int_{\mathbb{R}} e^{iyk} \sinh(pk) d_{t,x}(k) dk
= \int_\R e^{iyk} \frac{e^{pk}-e^{-pk}}{2} d_{t,x}(k)dk\\
= \frac{f(y-ip)-f(y+ip)}{2},
$$
and hence,
$$
d_{t,x}(k) = \frac{1}{\sinh(pk)} \frac{1}{2 \pi} \int_{\mathbb{R}} e^{-iyk} g_d(y) dy.
$$
\end{proof}

There exist several extensions of the Fourier pricing methods of Carr and Madan \cite{CarrMadan}.
For example, Lee \cite{Lee} shows that \eqref{fft} can be adjusted in order to allow for $p<0$
and derives pricing formulas for other European derivatives, such as binary or asset-or-nothing options.
By adjusting the proof of Theorem 5.1 of \cite{Lee} to our setup, one obtains the following: Denote
$$
a_{t,x}(k) = \Ex{e^{-R_t}S_t 1_{\{\log S_t>k\}}} \quad \mbox{and} \quad
b_{t,x}(k)= \Ex{e^{-R_t} 1_{\{\log S_t>k\}}}.
$$
Then
$$
c_{t,x}(k)=a_{t,x}(k)-e^kb_{t,x}(k),
$$
and one has

\begin{theorem}
Let $(t,x) \in \mathbb{R}_+ \times D$ and
$p,\;q\in\R$ such that $p+1,\;q\in U_{t,x}$. Then
\begin{align*}
a_{t,x}(k) = R_a(p) + \frac{e^{-p k}}{\pi}
\int_0^{\infty} \Re\left(e^{-iyk} g_a(y)\right) dy\\
b_{t,x}(k) = R_b(q) + \frac{e^{-q k}}{\pi}
\int_0^{\infty} \Re\left(e^{-iyk} g_b(y)\right) dy\\
c_{t,x}(k) = R_c(p) + \frac{e^{-p k}}{\pi}
\int_0^{\infty} \Re\left(e^{-iyk} g_c(y)\right) dy,
\end{align*}
where
\begin{align*}
g_a(y) = -\frac{h_{t,x}(p+1+iy)}{p+iy}, \quad
g_b(y) = \frac{h_{t,x}(q+iy)}{q+iy}, \quad
g_c(y) = \frac{h_{t,x}(p+1+iy)}{p^2 + p - y^2 + iy(2p+1)}
\end{align*}
and
$$
R_a(p) =
\begin{cases}
h_{t,x}(1) & \textrm{if }p<0\\
\frac{h_{t,x}(1)}{2} & \textrm{if }p=0\\
0 & \textrm{if }p>0
\end{cases}, \quad
R_b(q) =
\begin{cases}
h_{t,x}(0) & \textrm{if }q<0\\
\frac{h_{t,x}(0)}{2} & \textrm{if }q=0\\
0 & \textrm{if }q>0
\end{cases}
$$
$$
R_c(p) =
\begin{cases}
h_{t,x}(1)-e^kh_{t,x}(0) & \textrm{if }p<-1\\
h_{t,x}(1)-\frac{e^kh_{t,x}(0)}{2} & \textrm{if }p=-1\\
h_{t,x}(1) & \textrm{if }-1<p<0\\
\frac{h_{t,x}(1)}{2} & \textrm{if }p=0\\
0 & \textrm{if }p>0\\
\end{cases}.
$$
\end{theorem}

\subsection{Approximation of general payoffs}
\label{sub:approx}

\subsubsection{Idea}

The prices of European options with payoff functions in the set
$$
L_{t,x} := \crl{\varphi : \mathbb{R}_+ \rightarrow \mathbb{R}
\; \; \textrm{Borel-measurable such that } \Ex{e^{-R_t}|\varphi(S_t)|} < \infty}
$$
can be approximated by portfolios consisting of securities that can either be priced
directly or with Fourier methods. For most purposes it is sufficient to use a mix of
$\varphi(0)$ government bonds, power payoffs and call options. For instance,
let $p_1 < \dots < p_f$ be a set of powers and
$K_1 < \dots < K_g$ finitely many strike prices. Then fix $s^* > 0$ and determine
weights $v_1, \dots, v_f$ and $w_1, \dots, w_g$ by weighted $L^2$-regression:
\be \label{l2}
\mathop{\rm arg\,min}_{v,w}
\int_0^{s^*} \brak{\varphi(s) - \varphi(0)-\sum_{i=1}^f v_i s^{p_i} 1_{\crl{s > 0}} - \sum_{i=1}^g w_i (s-K_i)^+}^2 \rho(s) ds,
\ee
where $\rho$ is a heuristic density approximation of $S_t$. The positive function
$\rho$ is meant to put additional weight on regions where $S_t$ is more likely to lie
(usually in the vicinity of $\E_x \left[S_t\right]$ and if default is possible, around $0$).
If one does not have a good idea of the distribution of $S_t$, one
can also use non-weighted regression ($\rho\equiv1$).

If the integral is discretized, the optimization problem \eqref{l2} becomes a finite-dimensional
$L^2$-regression. To improve the numerical stability, one can first apply Gram--Schmidt orthogonalization
to the basis functions $\varphi(0)$, $s^{p_i} 1_{\crl{s > 0}}$ and $(s-K_i)^+$.

\subsubsection{Example: Truncated log payoff}

We illustrate this method by approximating the price of a truncated log payoff
$\varphi(s) = \log(s) \vee k$, $k \in \mathbb{R}$. Note that since
$S_t=0$ with positive probability, the truncation from below is crucial.

Assume $S_0=1$ and $k = -1$. We consider three ways of approximating
$\varphi$ with linear combinations of 101 instruments:
\begin{enumerate}
\item A government bond and power payoffs of powers $0.05,\;0.1,\dots,5$
\item A government bond and call options with strikes $0.03,\;0.06,\;\dots,3$
\item A government bond, power payoffs of powers $0.1,\;0.2,\dots,5$ and
call options with strikes $0.06,\dots,3$
\end{enumerate}
We let $s^*=3$. As heuristic density for $S_t$ we use $\rho$ given by
\begin{align}
\label{rho}
\rho(x)=\begin{cases}
\exp(-10x) & x<0.5\\
\exp(-10|x-1|) & 0.5\le x\le1.5\\
\exp(-5) & x>1.5
\end{cases}.
\end{align}
This choice of $\rho$ assigns more weight to $0$ and points around $1$.
But it is just an example of what one could choose. Depending on the model and the value of $t$
one may want to use different functions $\rho$.

Figure \ref{fig:approx} shows the errors of the three approximation methods. It can be
seen that for the truncated log payoff methods 2 and 3 give a much better approximation than method 1. 
The errors of methods 2 and 3 are similar. But since prices of power payoffs
are easier to calculate than those of call options, method 3 is significantly faster.

\begin{figure}
\centering
\includegraphics[scale=0.9]{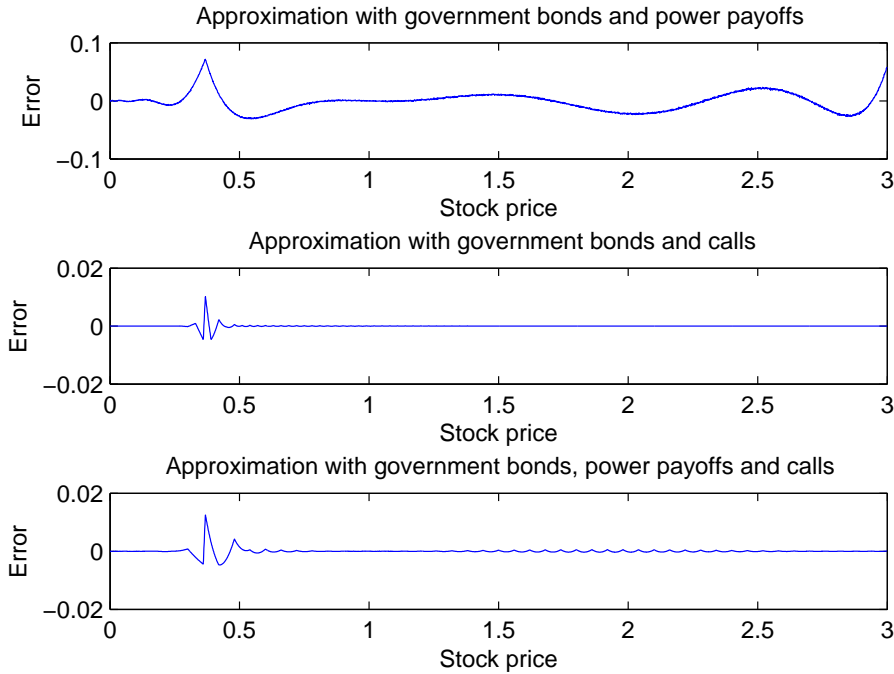}
\caption{Comparison of different approximation methods. Note that the scales are different!}
\label{fig:approx}
\end{figure}

Log payoffs are useful in the pricing and hedging of variance swaps on futures.
Let $F_u$, $0\le u\le t$, be the price of a futures contract on $S_t$ and consider
a variance swap with time-$t$ cash-flow
$$\Sigma_t=\max\left(\frac{1}{t}\sum_{i=1}^I \left(\log\frac{F_{t_i}}{F_{t_{i-1}}}\right)^2-K,C\right),
$$
where $0=t_0<\dots<t_I=t$ are the discrete monitoring points (usually daily),
$K$ is the strike and $C$ is a cap on the payoff (typically $C=2.5K$).
If $F_u$ is modeled as a positive continuous martingale
of the form $dF_u = \sigma_u F_u dW_u$, then the probability of hitting the
cap $C$ is negligible and one can approximate the sum with an integral:
$$
\Sigma_t \approx \frac{1}{t} \int_0^t \sigma^2_udu-K.
$$
It has been noticed by Dupire \cite{Dupire} and Neuberger \cite{Neuberger} that
$$
\int_0^t \sigma^2_u du = 2 \left(\int_0^t \frac{1}{F_u} dF_u - \log S_t + \log F_0 \right),
$$
and therefore,
$$
\Ex{\frac{1}{t} \int_0^t \sigma^2_u du} = \frac{2}{t} \Ex{\log F_0 - \log S_t}.
$$
In a diffusion model with the possibility of default the cap is crucial since it is hit in case
of default. If one neglects the probability that the cap is hit before default or
that $\log S_t < k :=  \log(F_0) - t(C+K)/2$ if there is no default, one can approximate
$\Sigma_t$ as follows:
\beas
\Sigma_t &\approx& 1_{\crl{\tau > t}} \brak{
\frac{1}{t} \int_0^t \sigma^2_u du - K} +1_{\crl{\tau \le t}} C\\
&\approx&
\frac{2}{t} \brak{\int_0^{t \wedge \tau} \frac{1}{F_{u-}} dF_u - (\log (S_t) \vee k) + \log F_0} - K
- 1_{\crl{\tau \le t}} \frac{2}{t} \int_0^{\tau} \frac{1}{F_{u-}} dF_u.
\eeas
In the special case where the default intensity is constant, the expectation of the
last integral is zero, and one can price according to
$$
\Ex{\Sigma_t} \approx \frac{2}{t} \Ex{\log F_0 - \log (S_t \vee k)} - K.
$$

\setcounter{equation}{0}
\section{Hedging}
\label{sec:compl}

In this section we consider a subclass of affine models in which
European contingent claims can perfectly be hedged by dynamically trading in sufficiently many
liquid securities. Assume that condition
\eqref{martcond} is fulfilled. So that the discounted stock price
$\exp(s_t + \Lambda_t) 1_{\crl{t < \tau}}$ is a martingale under all $\p_x$, $x \in D$.
It is well-known that it is impossible to replicate contingents claims with
finitely many hedging instruments
in a model where the underlying has jumps of infinitely many different sizes.
Therefore, we here require the jump measures $\nu$ and $\mu_i$, $i \in {\cal I}$,
to be of the following form:
\bea
\label{nuform}
\nu &=& \sum_{q=1}^{M} v_q \delta_{y_q} \mbox{ for } v_q>0 \mbox{ and different points }
y_1, \dots, y_M \mbox{ in } \in D \backslash \{0\}\\
\label{muform}
\mu_i &=& \sum_{q=1}^{M_i} v_{iq} \delta_{y_{iq}} \mbox{ for } v_{iq}>0
\mbox{ and different points } y_{i1}, \dots, y_{iM_i} \mbox{ in }\in D \backslash\{0\},
\eea
where $\delta_{y_q}$, $\delta_{y_{iq}}$ are Dirac measures and
for $M = 0$ or $M_i = 0$, $\nu$ or $\mu_i$ are understood to be zero, respectively.
The integrability condition (vii) is then trivially satisfied.

In the following theorem we are going to
show that the process $X$ has a realization as the unique strong solution of an SDE of the from
\be \label{SDEjumps}
dX_t = \tilde{b}(X_t) dt + \sigma(X_t) dW_t + \int_{\mathbb{R}_+} k(X_{t-},z) N(dt,dz),
\ee
where $W$ is an $n$-dimensional Brownian motion and $N$ an independent
Poisson random measure on $\mathbb{R}^2_+$ with Lebesgue measure as intensity measure.
$\sigma : D \to \mathbb{R}^{n \times n}$ has to be
a measurable function satisfying
\be \label{sigma}
\sigma \sigma^T(x) = a + \sum_{i \in {\cal I}} x_i \alpha_i \quad \mbox{for all $x \in D,$}
\ee
and the functions $\tilde{b} : D \to \mathbb{R}^n$ and
$k : D \times \mathbb{R}_+ \to D$ are of the following form:
$$
\tilde{b}(x) = b + \beta x - \int_{D \setminus \crl{0}} \chi_{\cal J}(\xi) \nu(d \xi)
- \sum_{i=1}^m x_i \int_{D \setminus \crl{0}} \chi_{{\cal J} \cup \crl{i}}(\xi) \mu_i(d\xi),
$$
$$
k(x,z) = \left\{
\begin{array}{ll}
y_q & \mbox{ if } V_{q-1} \le z < V_q\\
y_{iq} & \mbox{ if } V_M + \sum_{j=1}^{i-1} x_j V_{jM_j} + x_i V_{i(q-1)}
\le z < V_M + \sum_{j=1}^{i-1} x_j V_{j M_j} + x_i V_{iq}\\
0 & \mbox{ if } z \ge V_M + \sum_{j=1}^{m} x_j V_{jM_j}
\end{array} \right.,
$$
where $V_q = \sum_{p=1}^q v_p$ and $V_{iq} = \sum_{p=1}^q v_{ip}$.
Theorem \ref{thmSDEjumps} is an extension of Theorem 8.1 in
Filipovic and Mayerhofer \cite{FilipovicMayerhofer}. Its proof is given in the appendix.

\begin{theorem} \label{thmSDEjumps}
If $\nu$ and $\mu_i$, $i \in {\cal I}$, are of the form \eqref{nuform}--\eqref{muform},
then there exists a measurable function $\sigma : D \to \mathbb{R}^{n \times n}$
satisfying \eqref{sigma} such that the SDE \eqref{SDEjumps} has for all
initial conditions $x \in D$ a unique strong solution. It is Feller process
satisfying {\rm (H1)--(H2)} corresponding to the
parameters $(a,\alpha,b,\beta,m,\mu)$.
\end{theorem}

In general, one needs $L+1 := n+M+\sum_{i=1}^m M_i+2$ instruments to hedge all sources of risk.
We assume the first one to be a money market account yielding an instantaneous return of $r_t$.
In addition, one needs one instrument for each of the $n$ components of the
Brownian motion $W$, one for each of the possible jumps of $X$ and
one for the jump to default (without default, $L$ instruments are generally enough).
Of course, if the hedging instruments are redundant, not all
European contingent claims can be replicated; precise conditions
for a given option to be replicable are given in \eqref{eq_comp} below.
The hedging instruments should also be liquidly traded. For instance,
in addition to the money market account,
one could use a mix of instruments of the following types:
\begin{itemize}
\item Stock shares
\item Government bonds
\item Corporate bonds
\item Vanilla options.
\end{itemize}
All of these can be viewed as European options with different payoff functions.
Let us denote the set of hedging instruments different from the money market account by
$$
\Phi = \{(t_1,\varphi_1),\dots,(t_L,\varphi_L)\},$$
where $t_1,\dots,t_L$ are the maturities and $\varphi_l \in L_{t_l,x}$ the payoff functions.

Now consider a European option with maturity $t \le \min \crl{t_1, \dots, t_L}$ and
payoff function $\varphi\in L_{t,x}$. At time $0$, its price is
$$
\pi(t,x) := \E_x \edg{\exp(-R_t)\varphi(S_t)},
$$
and for $u \in [0,t]$,
$$
C_u = \left\{
\begin{array}{cl}
\pi(t-u,X_u) & \mbox{ if } u < \tau\\
\E_x \edg{e^{-(R_t-R_u)} \mid X_u} \varphi(0) & \mbox{ if } u \ge \tau
\end{array} \right. .
$$
After the default time $\tau$, $C_u$ behaves like a government bond and
can be hedged accordingly. To design the hedge before time $\tau$,
we introduce the following sensitivity parameters:
\begin{itemize}
\item Classical Greeks: for $q=1,\dots,n$:
$$H^q_{t,x} =\frac{\partial}{\partial x_q} \pi(t,x)$$
\item Sensitivities to the jumps corresponding to $\nu$: for all $q = 1, \dots, M$:
$$J^q_{t,x} = \pi(t,x+y_q) - \pi(t,x)$$
\item Sensitivities to the jumps corresponding to $\mu$:
for all $i \in\mathcal{I}$ and $q = 1, \dots,M_i$:
$$J^{iq}_{t,x} =\pi(t,x+y_{iq}) - \pi(t,x)$$
\item Sensitivity to default:
$$D_{t,x} =\Ex{\exp(-R_t)\varphi(0)}-\Ex{\exp(-R_t)\varphi(S_t)}.$$
\end{itemize}

\begin{Example}
\label{H1}
Consider a power payoff $\varphi(s)=s^p$ for some $p > 0$ such that $s^p \in L_{t,x}$.
Then the sensitivity parameters are given by
\beas
H^q_{t,x} &=& \frac{\partial}{\partial x_q}h_{t,x}(p) = B_q(t,p\varepsilon,p-1,p)h_{t,x}(p)\\
J^{q}_{t,x} &=& h_{t,x+y_q}(p)-h_{t,x}(p)\\
J^{iq}_{t,x} &=& h_{t,x+y_{iq}}(p)-h_{t,x}(p)\\
D_{t,x} &=& -h_{t,x}(p).
\eeas
\end{Example}

\begin{Example}
\label{H2}
For a European call option with log strike $k$ and maturity $t$, the classical Greeks are
\begin{align*}
H^q_{t,x} = \frac{\partial}{\partial x_q} c_{t,x}(k) = \frac{e^{-p k}}{\pi}
\int_0^{\infty} \Re \brak{e^{-iyk} \partial_{x_q} g_c(y)} dy,
\end{align*}
where
\beas
\partial_{x_q} g_c(y) &=& \partial_{x_q} \frac{h_{t,x}(p+1+iy)}{p^2 + p - y^2 + iy(2p+1)}\\
&=& \frac{B_q(t,(p+1+iy) \varepsilon,p+iy,p+1+iy)h_{t,x}(p+1+iy)}{p^2 + p - y^2 + iy(2p+1)}
\eeas
for some $p > 0$ such that $1+p \in U_{t,x}$.
The other sensitivities are
\begin{eqnarray*}
J^{q}_{t,x} &=& c_{t,x+y_q}(k)-c_{t,x}(k)\\
J^{iq}_{t,x} &=& c_{t,x+y_{iq}}(k)-c_{t,x}(k)\\
D_{t,x} &=& -c_{t,x}(k).
\end{eqnarray*}
\end{Example}

As shown in Examples \ref{H1} and \ref{H2}, the sensitivity parameters of power payoffs
and vanilla options can be given in closed or almost closed form.
The payoff $\varphi \in L_{t,x}$ can be approximated with
a linear combination of government bonds, power payoffs
and European calls of maturity $t$ as in Subsection \ref{sub:approx}.
The sensitivities of $\varphi$ are then approximated by the
the same linear combination of the sensitivities of the government bond, power payoffs and
European calls. If $X$ is a solution of the SDE \eqref{SDEjumps} and $\pi$ is a
$C^{1,2}$-function on $(0,t] \times D$, one obtains from \Ito's formula that for
$u \le t \wedge \tau$,
\beas
dC_u&=&\textrm{``drift"}\;du+\sum_{q,q'=1}^n H^q_{t-u,X_{u-}} \sigma_{qq'}(X_{u-})dW^{q'}_u
+ D_{t-u,X_{u-}} d1_{\crl{\tau \le u}}\\
&& + \int_{\mathbb{R}_+} (\pi(t-u,X_{u-} + k(X_{u-},z)) - \pi(t-u,X_{u-})) dN(du, dz),
\eeas
which can be written as
\beas
dC_u &=& \textrm{``drift"} \; du+\sum_{q,q'=1}^n H^q_{t-u,X_{u-}} \sigma_{qq'}(X_{u-})dW^{q'}_u
+ D_{t-u,X_{u-}} d1_{\crl{\tau \le u}}\\
&& + \sum_{q=1}^M J^q_{t-u,X_{u-}} dN^q_u
+ \sum_{i=1}^m \sum_{q=1}^{M_i} J^{iq}_{t-u,X_{u-}} dN^{iq}_u,
\eeas
where $N^q$ and $N^{iq}$ are Poisson processes with stochastic intensity
depending on $X$ but independent of the payoff function $\varphi$. For all $l=1, \dots, L$, define
$\pi^l$ and $C^l$ analogously to $\pi$ and $C$, and assume the functions $\pi^l$ are all
$C^{1,2}$ on $(0,t] \times D$. To hedge before default, one has to invest in
the hedging instruments such that the
resulting portfolio has the same sensitivities. That is, one tries to find
$\vartheta(u,x) \in \mathbb{R}^L$ such that for all $0 \le u < t \wedge \tau$,
\begin{align}
\label{eq_comp}
\begin{cases}
H^1_{t-u,X_u} = \sum_{l=1}^L \vartheta^l(u,X_u) H^{l,1}_{t_l-u,X_u}\\
\vdots\\
H^n_{t-u,X_u} = \sum_{l=1}^L \vartheta^l(u,X_u) H^{l,n}_{t_l-u,X_u}\\ \\
J^1_{t-u,X_u}=\sum_{l=1}^L \vartheta^l(u,X_u) J^{l,1}_{t_l-u,X_u}\\
\vdots\\
J^M_{t-u,X_u}=\sum_{l=1}^L \vartheta^l(u,X_u) J^{l,M}_{t_l-u,X_u}\\ \\
J^{11}_{t-u,X_u}=\sum_{l=1}^L \vartheta^l(u,X_u) J^{l,11}_{t_l-u,X_u}\\
\vdots\\
J^{m M_m}_{t-u,x}=\sum_{l=1}^L \vartheta^l(u,X_u) J^{l,mM_m}_{t_l-u,X_u}\\ \\
D_{t-u,X_u} =\sum_{l=1}^L \vartheta^l(u,X_u) D^l_{t_l-u,X_u},
\end{cases}
\end{align}
where on the left side are the sensitivities of $\varphi(S_t)$ and on the right, indexed by $l$,
the sensitivities of the hedging instruments $\varphi^l(S_{t_l})$.
$\vartheta^l(u,X_u)$ is the number of hedging instrument $l$ in the hedging portfolio before
default while the amount $C_u - \sum_{l=1}^L \vartheta^l(u,X_u) C^l_u$
is held in the money market account. Since $e^{-R_u} C_u$ and
$e^{-R_u} C^l_u$, $l =1, \dots, L$ are all martingales under $\p_x$,
this is a self-financing strategy replicating $C$ until time
$t \wedge \tau$. If default happens before time $t$, one simply holds
$\varphi(0)$ zero coupon government bonds from then until time $t$.

\eqref{eq_comp} is a system of $L$ linear equations with $L$ unknowns.
It may not have a solution if it is degenerate. But if the family $\Phi$
of hedging instruments is such that \eqref{eq_comp} has full rank for all
$0 \le u \le t \wedge \tau$, then any
European contingent claim can be replicated by dynamic trading.

\setcounter{equation}{0}
\section{Heston model with stochastic interest rates and possibility of default}
\label{sec:app}

As an example we discuss a Heston-type stochastic volatility model with stochastic interest rates and
possibility of default. It extends the model of Carr and Schoutens \cite{CarrSchoutens} and can
easily be extended further to include more risk factors.
Let $(X_t)_{t\ge0}$ be a process with values in $D = \mathbb{R}_+^2 \times \mathbb{R}$
moving according to
\bea 
\label{Heston1}
dX^1_t &=& \kappa_1\left(\theta_1-X^1_t\right)dt+ \eta_1 \sqrt{X^1_t}dW^1_t\\
\label{Heston2}
dX^2_t &=& \kappa_2\left(\theta_2-X^2_t\right)dt+ \eta_2 \sqrt{X^2_t}dW^2_t\\
\label{Heston3}
dX^3_t &=& - \frac{1}{2} X^1_t dt + \sqrt{X^1_t}dW^3_t
\eea
for a 3-dimensional Brownian motion $W$ with correlation matrix
$$\left(\begin{array}{ccc}
1 & 0 & \rho\\
0 & 1 & 0\\
\rho & 0 & 1
\end{array}\right)
$$
and non-negative constants $\kappa_1$, $\kappa_2$, $\theta_1$, $\theta_2$, $\eta_1$, $\eta_2$.
Since $X^1$ and $X^2$ are autonomous square-root processes, the system 
\eqref{Heston1}--\eqref{Heston3} has for all initial conditions $x \in \mathbb{R}^2_+ \times \mathbb{R}$
a unique strong solution, and it follows as in
the proof of Theorem \ref{thmSDEjumps} that it is a Feller process satisfying (H1)--(H2)
with parameters
\begin{itemize}
\item $a=0$, \; $\alpha^1= \frac{1}{2} \left(
\begin{array}{ccc}
\eta_1^2 & 0 & \rho \eta_1\\
0 & 0 & 0\\
\rho \eta_1 & 0 & 1
\end{array}\right)$, \; $\alpha^2= \frac{1}{2} \left(
\begin{array}{ccc}
0 & 0 & 0\\
0 & \eta_2^2 & 0\\
0 & 0 & 0
\end{array}\right)$\\
\item $b=(\kappa_1\theta_1,\kappa_2\theta_2,0)$, \; $\beta=\left(
\begin{array}{ccc}
-\kappa_1 & 0 & 0\\
0 & -\kappa_2 & 0\\
- 1/2 & 0 & 0
\end{array}\right)$.
\end{itemize}
Note that one cannot have correlation between $W^1$ and $W^2$ or $W^2$ and $W^3$ 
without destroying the affine structure of $X$. So we introduce 
dependence between $s$, the volatility process $X^1$, $r$ and $\lambda$ by setting
\beas
s_t &=& X^3_t, \quad r_t = d +\delta_1 X^1_t+\delta_2X^2_t,
\quad \lambda_t = c + \gamma_1 X^1_t + \gamma_2 X^2_t,\\
R_t &=& \int_0^t r_u du, \quad \Lambda_t = \int_0^t \lambda_u du
\eeas
for non-negative constants $c, \gamma_1, \gamma_2, d, \delta_1, \delta_2$. Then
$$
S_t = \exp \brak{s_t + R_t + \Lambda_t} 1_{\crl{t < \tau}}
$$
satisfies the SDE
$$
dS_t= S_{t-} \brak{\edg{c+d+(\gamma_1+\delta_1) X^1_t + \left(\gamma_2+\delta_2\right) X^2_t}
dt + \sqrt{X^1_t} dW^3_t- d1_{\crl{\tau \le t}}}.$$

\subsection{Pricing equations and moment explosions}

It follows from Corollary \ref{mp} that the discounted price
$\exp \brak{s_t + \Lambda_t} 1_{\crl{t < \tau}}$ is a martingale.
The discounted moments function is of the form
\beas
h_{t,x}(z)&=&\exp(A(t,(0,0,z),z-1,z)+\left<B(t,(0,0,z),z-1,z),x\right>)\\
&=&\exp\left(\tilde{A}(t,z)+\tilde{B}_1(t,z)x_1+\tilde{B}_2(t,z)x_2+ z x_3 \right),
\eeas
where
\begin{align}
\label{RicModel}
\begin{cases}
\partial_t \tilde{A}(t,z)=\kappa_1\theta_1\tilde{B}_1(t,z)+\kappa_2 \theta_2\tilde{B}_2(t,z)+(c+d)(z-1)\\
\partial_t \tilde{B}_1(t,z)=\frac{1}{2}\eta_1^2\tilde{B}_1^2(t,z)+(\rho\eta_1z-\kappa_1)\tilde{B}_1(t,z)
+ (\frac{1}{2}z + \gamma_1 + \delta_1) (z-1)\\
\partial_t \tilde{B}_2(t,z)=\frac{1}{2}\eta_2^2\tilde{B}_2^2(t,z)-\kappa_2\tilde{B}_2(t,z)+(\gamma_2+\delta_2)(z-1)\\[2mm]
\tilde{A}(0,x)=\tilde{B}_1(0,z)=\tilde{B}_2(0,z)=0.
\end{cases}
\end{align}
In this special case, $\tilde{B}_1$ and $\tilde{B}_2$ are both solutions of scalar Riccati ODEs that can be obtained
explicitly. The explosion times of the discounted moments
$$
t^*(p):=\sup\left\{t\ge0:\E[\exp(-R_t)S^p_t 1_{\crl{t < \tau}}]<\infty\right\}
$$
can also be determined in closed form:
\begin{align}
\label{tstar}
t^*(p)=
\begin{cases}
\frac{2}{\sqrt{-u}}\left(\arctan \frac{-u}{v}+\pi 1_{\{v<0\}}\right)&\textrm{if } u<0\\
\frac{1}{\sqrt{u}}\log\frac{v+\sqrt{u}}{v-\sqrt{u}}&\textrm{if } u \ge0,\;v >0\\
\infty & \textrm{otherwise}
\end{cases},
\end{align}
where
\beas
u &=& \rho \eta_1 p -\kappa_1,\\
v &=& \left(\frac{1}{2}p + \gamma_1 + \delta_1 \right)(p-1).
\eeas
The derivation of \eqref{tstar} is analogous to the derivation of the same formula
for the Heston model which can be found, for example, in Andersen and Piterbarg \cite{AndersenPiterbarg}.
Figure \ref{fig:explosion} shows the decay
of $t^*(p)$ for $p\rightarrow \pm \infty$.
\begin{figure}[htp]
\centering
\includegraphics[scale=.9]{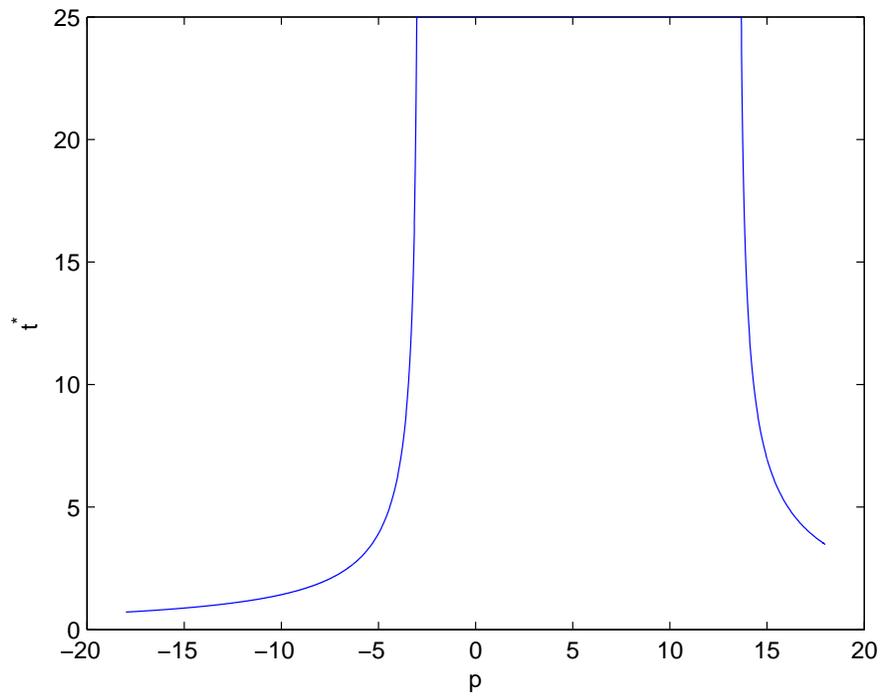}
\caption{Moment explosion times ($\kappa_1=0.06$, $\kappa_2=0.04$, $\eta_1=0.2$, $\eta_2=0.1$,
$\theta_1=1$, $\theta_2=0.3$, $\rho=-0.6$, $\gamma_0+\delta_0=0.02$, $\gamma_1+\delta_1=0.2$, $\gamma_2+\delta_2=0.2$).
}\label{fig:explosion}
\end{figure}
A sample plot of the implied volatility surfaces for the two cases of positive
probability of default and no default is given in Figure \ref{fig:impvol1}.
For the calculation of implied volatilities we used the yield on a government bond
as the interest rate in the Black--Scholes formula.
As one would expect, credit risk contributes towards a higher implied volatility,
especially at longer maturities or at extreme strikes. This effect can help explain why
implied volatilities usually exceed realized volatilities.
\begin{figure}
\centering
\includegraphics[scale=.8]{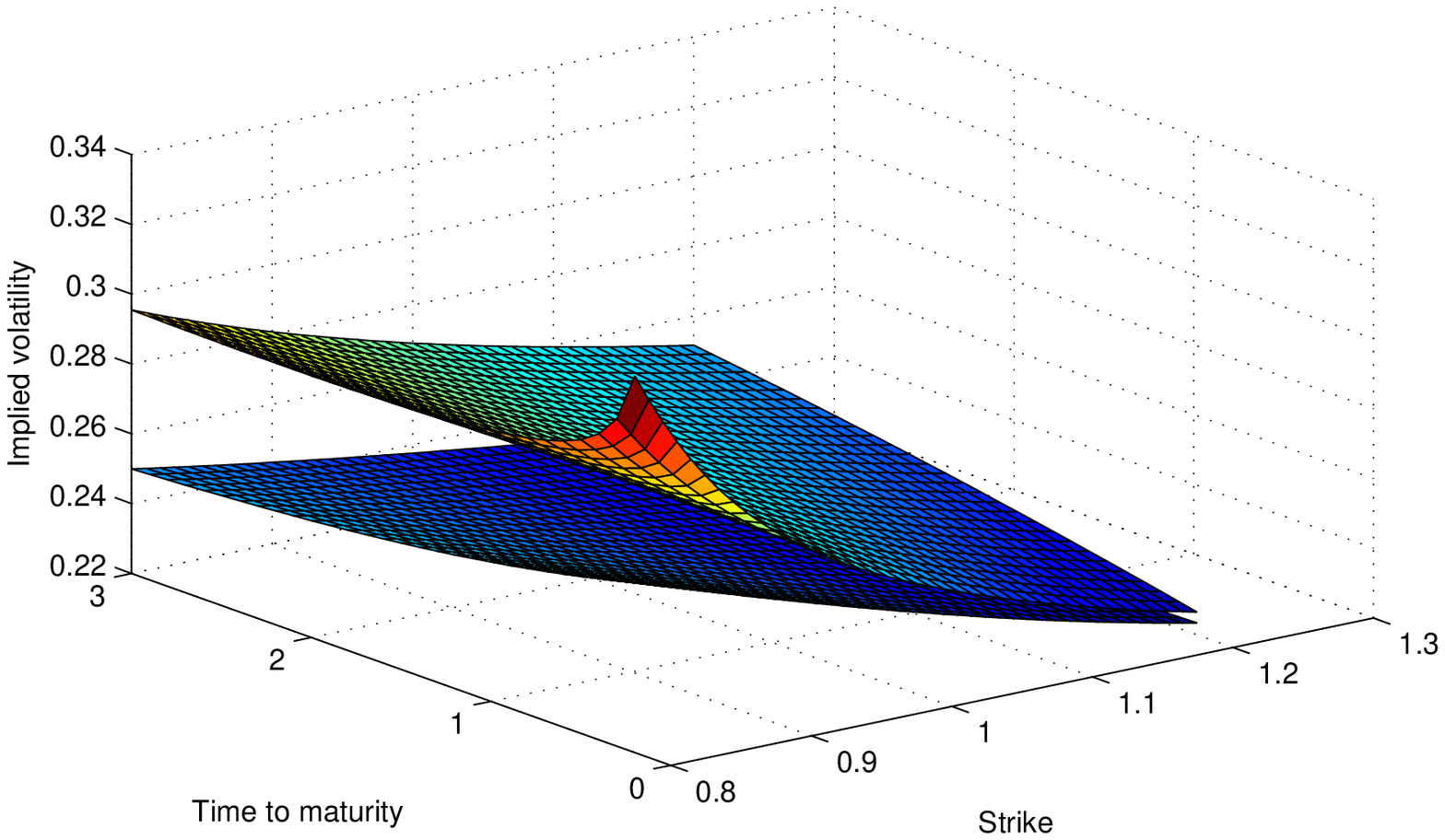}
\caption{Implied volatility surfaces with positive probability of default (upper
surface) vs. no default (lower surface). Parameters: $\kappa_1=0.06$,
$\kappa_2=0.04$, $\eta_1=0.2$,
$\eta_2=0.1$, $\theta_1=1$, $\theta_2=0.3$, $\rho=-0.6$, $c=0.02$ (no
default: $c=0$),
$\gamma_1=0.01$ (no default: $\gamma_1=0$), $\gamma_2=0.01$ (no
default: $\gamma_2=0$), $d=0.01$, $\delta_1=0.1$,
$\delta_2=0.1$, $X^1_0=0.05$, $X^2_0=0.03$, $X^3_0 =0$.}
\label{fig:impvol1}
\end{figure}

\subsection{Hedging}
\label{appl_hedging}
Denote by $\tilde{A}^0(t,z)$, $\tilde{B}^0_1(t,z)$ and $\tilde{B}^0_2(t,z)$
the solutions of the Riccati equations
\eqref{RicModel} for $c = \gamma_1 = \gamma_2=0$ (no default). Then
the price of a zero coupon government bond with maturity $t$ is given by
$$P^0_{t,x}=\exp(\tilde{A}^0(t,0)+\tilde{B}^0_1(t,0)x_1+\tilde{B}^0_2(t,0)x_2).$$
Assume there exists a $t > 0$ such that
\bea
\label{techCond}
&& \partial_{x_1} c_{u,x}(k)\tilde{B}^0_2(u,0)-\partial_{x_2}c_{u,x}(k)\tilde{B}^0_1(u,0) \neq\\
\notag
&& [\tilde{B}^0_1(u,0)\tilde{B}_2(u,0)-\tilde{B}^0_2(u,0)\tilde{B}_1(u,0)][\partial_{x_3}
c_{u,x}(k)-c_{u,x}(k)]
\eea
for all $0 \le u \le t$ and $x \in D$. Then every European contingent claim can be
hedged if the following four instruments can liquidly be traded: the stock, a zero
coupon government bond, a zero coupon corporate bond and a call option with log
strike $k$, the latter three all with maturity $t$.

Indeed, by Remarks \ref{H1} and \ref{H2}, the hedging parameters are as follows:
\begin{itemize}
\item Stock: $H^{S,1}_{u,x}= H^{S,2}_{u,x}=0$, $H^{S,3}_{u,x}=e^{x_3}$, $D^S_{u,x}=-e^{x_3}$.
\item Government bond: $H^{G,1}_{u,x}= \tilde{B}_1^0(u,0)P^0_{u,x}$,
$H^{G,2}_{u,x}= \tilde{B}_2^0(u,0) P^0_{u,x}$, $H^{G,3}_{u,x}=0$, $D^G_{u,x}=0$.
\item Corporate bond: $H^{B,1}_{u,x}= \tilde{B}_1(u,0)P_{u,x}$,
$H^{B,2}_{u,x}= \tilde{B}_2(u,0)P_{u,x}$,
$H^{B,3}_{u,x}=0$, $D^B_{u,x}= - P_{u,x}$.
\item Call option: $H^{C,q}_{u,x} = \partial_{x_q} c_{u,x}(k)$, $q = 1,2,3$, $D^{C}_{u,x}=-c_{u,x}(k)$,
\end{itemize}
and \eqref{techCond} is equivalent to
\beas
&& e^{x_3}P^0_{u,x}P_{u,x} \{-\partial_{x_1}c_{u,x}(k)
\tilde{B}^0_2(u,0)+\partial_{x_2}c_{u,x}(k)\tilde{B}^0_1(u,0)\\
&&+[\tilde{B}^0_1(u,0)\tilde{B}_2(u,0)-\tilde{B}^0_2(u,0)\tilde{B}_1(u,0)]
[\partial_{x_3}c_{u,x}(k)-c_{u,x}(k)]\} \neq 0.
\eeas
So the matrix
$$
\left(
\begin{array}{cccc}
0 & \tilde{B}^0_1(u,0)P^0_{u,x} & \tilde{B}_1(u,0)P_{u,x} & \partial_{x_1} c_{u,x}(k)\\
0 & \tilde{B}^{0}_2(u,0)P^0_{u,x} & \tilde{B}_2(u,0)P_{u,x} & \partial_{x_2} c_{u,x}(k)\\
e^{x_3} & 0 & 0 & \partial_{x_3} c_{u,x}(k)\\
-e^{x_3} & 0 & -P_{u,x} & -c_{u,x}(k)
\end{array}
\right).
$$
has full rank, and the system \eqref{eq_comp} always has a solution.

\begin{appendix}

\setcounter{equation}{0}
\section{Proof of Theorem \ref{thmSDEjumps}}

It is shown in Theorem 8.1 of
Filipovic and Mayerhofer \cite{FilipovicMayerhofer} that there exists a measurable function
$\sigma : D \to \mathbb{R}^{n \times n}$ satisfying \eqref{sigma} such that the SDE
\be \label{contSDE}
dX_t = \tilde{b}(X_t) dt + \sigma(X_t) dW_t
\ee
has for each initial condition $x \in D$ a unique strong solution $X^{(0)}$. Now set
$\tau_0 := 0$ and define iteratively
\bea
\label{tauq}
\tau_q &:=& \inf \crl{t > \tau_{q-1} : \int_0^t\int_{\R_+} k(X^{(q-1)}_{u-}, z) N(du,dz) \neq 0}\\
\notag
X^{(q)}_t &:=& 1_{\crl{0 \le t < \tau_q}} X^{(q-1)}_t
+ 1_{\crl{\tau_q \le t < \infty}} Y^{(q)}_t,
\eea
where $Y^{(q)}$ is the solution of the SDE \eqref{contSDE} on $[\tau_q, \infty)$ with
initial condition
$$
Y^{(q)}_{\tau_q} = X^{(q-1)}_{\tau_q} + k(X^{(q-1)}_{\tau_q-}, z) N(\tau_q,dz).
$$
Since $X^{(q-1)}$ is RCLL and the intensity measure of $N$ is Lebesgue measure,
it can be seen from \eqref{tauq} that $\tau_q > \tau_{q-1}$ a.s. So the process
$$
X^{(\infty)}_t := \sum_{q \ge 1} 1_{\crl{\tau_{q-1} \le t < \tau_q}} X^{(q-1)}_t
$$
is the unique strong solution of \eqref{SDEjumps} on $[0,\tau_{\infty})$, where
$\tau_{\infty} := \lim_{q \to \infty} \tau_q$.
It remains to show that $\tau_{\infty} = \infty$ and $X^{(\infty)}$ is a Feller process satisfying
(H1)--(H2) with parameters $(a,\alpha, b, \beta,m,\mu)$.
To do that we introduce the counting process
$$
Z_t := \sum_{q \ge 1} 1_{\crl{t \ge \tau_q}}.
$$
By \Ito's formula,
$$
f (X^{(\infty)}_{t \wedge \tau_q}, Z_{t \wedge \tau_q})
- \int_0^{t \wedge \tau_q}
\tilde{\mathcal{A}} f (X^{(\infty)}_u ,Z_u)du
$$
is for all $f \in C^2_c(D \times \mathbb{R}_+)$ and $q \ge 1$, a martingale, where
\beas
\tilde{\mathcal{A}}f(x,z) &=& \sum_{k,l=1}^n \left(a_{kl}+\left<\alpha_{\mathcal{I},kl},
x_\mathcal{I}\right>\right)\frac{\partial^2 f(x,z)}{\partial x_k\partial x_l}
+ \sum_{k=1}^n \brak{b_k + \sum_{l=1}^n \beta_{kl} x_l} \frac{\partial f(x,z)}{\partial x_k}\\
&& +\int_{D \backslash\{0\}} \left(f(x+\xi,z+1)-f(x,z)-\left<\nabla_{\mathcal{J}} f(x,z),
\chi_{\mathcal J}(\xi)\right>\right) \nu(d\xi)\\
&& +\sum_{i=1}^m \int_{D\backslash\{0\}}
\left(f(x+\xi,z+1)-f(x,z)-\left<\nabla_{\mathcal{J}\cup\{i\}}f(x,z),\chi_{\mathcal{J}
\cup\{i\}}(\xi)\right>\right)x_i\mu_i(d\xi).
\eeas
But $\tilde{\cal A}$ is the infinitesimal generator of a regular affine process $\tilde{X}$ 
with values in $D \times \mathbb{R}_+$.
So by Theorem 2.7 of Duffie et al. \cite{DFS},
$C^2_c(D \times \mathbb{R}_+)$ is a core of $\tilde{\cal A}$, and it follows from
Theorem 4.4.1 of Ethier and Kurtz \cite{EthierKurtz} that the martingale problem for
$\tilde{\mathcal{A}} \mid_{C^2_c(D \times \mathbb{R}_+)}$ is well-posed.
Moreover, the stopping times $\tau_q$ are exit times:
$$
\tau_q = \inf \crl{t\ge 0 : Z_t \notin [0, q - 1/2)}.
$$
Therefore, one obtains from Theorem 4.6.1 of Ethier and Kurtz \cite{EthierKurtz} that
the stopped martingale problem corresponding to $\tilde{\cal A} \mid_{C^2_c(D \times \mathbb{R}_+)}$
and $D \times [0, q - 1/2)$ is well-posed for all $q$. Hence,
$(X^{(\infty)}_{t \wedge \tau_q},Z_{t \wedge \tau_q})_{t \ge 0}$ has the same distribution as
$(\tilde{X}_{t \wedge \tilde{\tau}_q})_{t \ge 0}$, where $\tilde{\tau}_q$ is the $q$-th jump time of $\tilde{X}$.
Since $\tilde{X}$ is RCLL and $V_M + \sum_{i=1}^m V_{iM_i} X_{i,t}^{(\infty)}$ is the jump intensity of
$Z$, we conclude that almost surely, the process $X^{(\infty)}_{t \wedge \tau_{\infty}}$
jumps at most finitely many times on compact time intervals.
In particular, $\tau_{\infty} = \infty$, and Theorem \ref{thmSDEjumps} follows from 
Theorem 2.7 of Duffie et al. \cite{DFS} since the first $n$ components of $\tilde{X}$ form a 
regular affine processes with infinitesimal generator ${\cal A}$ acting on $f \in C^2_c(D)$ like
\beas
\mathcal{A}f(x) &=& \sum_{k,l=1}^n \left(a_{kl}+\left<\alpha_{\mathcal{I},kl},
x_\mathcal{I}\right>\right)\frac{\partial^2 f(x)}{\partial x_k\partial x_l}
+ \left<b+\beta x, \nabla f(x)\right>\\
&& +\int_{D \backslash\{0\}} \left(f(x+\xi)-f(x)-\left<\nabla_{\mathcal{J}} f(x),
\chi_{\mathcal J}(\xi)\right>\right) \nu(d\xi)\\
&& +\sum_{i=1}^m \int_{D\backslash\{0\}}
\left(f(x+\xi)-f(x)-\left<\nabla_{\mathcal{J}\cup\{i\}}f(x),\chi_{\mathcal{J}
\cup\{i\}}(\xi)\right>\right)x_i\mu_i(d\xi).
\eeas

\end{appendix}


\begin{thebibliography}{25}
\bibitem{AndersenPiterbarg}{Andersen L., Piterbarg V. (2007).
Moment explosions in stochastic volatility models. Finance and Stochastics, 11, 29--50.}

\bibitem{CarrMadan}{Carr P., Madan D. (1999). Option valuation using the fast Fourier transform.
Journal of Computational Finance 2(4), 61--73.}

\bibitem{CarrSchoutens}{Carr P., Schoutens W. (2008). Hedging under the Heston model with jump-to-default.
International Journal of Theoretical and Applied Finance 11(4), 403--414.}

\bibitem{CIR}{Cox J., Ingersoll J., Ross S. (1985). A theory of the term structure of interest rates.
Econometrica 53, 385--467.}

\bibitem{Dupire}{Dupire B. (1993). Model art. Risk 6(9), 118--124.}

\bibitem{DFS}{Duffie D., Filipovic D., Schachermayer W. (2003). Affine processes and applications in finance.
Annals of Applied Probability 13(3), 984--1053.}

\bibitem{EthierKurtz}{Ethier S., Kurtz T. (2005). Markov Processes: Characterization and Convergence. John Wiley \& Sons.}

\bibitem{FilipovicMayerhofer}{Filipovic D., Mayerhofer E. (2009). Affine diffusion processes:
Theory and applications. Advanced Financial Modelling, volume 8 of Radon
Series on Computational and Applied Mathematics, Walter de Gruyter, Berlin.}

\bibitem{Heston}{Heston S. (1993). A closed form solution for options with stochastic
volatility with applications to bond and currency options. Review of Financial Studies 6(2), 327--343.}

\bibitem{IY}{Ilyashenko Y., Yakovenko S. (2008). Lectures on Analytic Differential Equations.
American Mathematical Society. Graduate Studies in Mathematics 86.}

\bibitem{KR}{Keller-Ressel M. (2011). Moment explosions and long-term behavior of
affine stochastic volatility models. Mathematical Finance, 21(1), 73--98.}

\bibitem{KR2}{Keller-Ressel M., Teichmann J., Schachermayer W. (2010). Affine processes are regular.
Probability Theory and Related Fields, 1--21.}

\bibitem{Lando}{Lando D. (1998). On Cox processes and credit risky securities.
Review of Derivatives Research 2, 99--120.}

\bibitem{MomentFormula}{Lee R. (2004). The moment fomula for implied volatility at extreme strikes.
Mathematical Finance 14(3), 469--480.}

\bibitem{Lee}{Lee R. (2005). Option pricing by transform methods: extensions,
unification, and error control. Journal of Computational Finance 7(3), 51--86.}

\bibitem{Lukacs}{Lukacs E. (1970). Characteristic Functions. Charles Griffin \& Co., Second Edition.}

\bibitem{Neuberger}{Neuberger A. (1994). The log contract. Journal of Portfolio Management 20(2), 74--80.}

\bibitem{SpreijVeerman}{Spreij P., Veerman E. (2010).
The affine transform formula for affine jump-diffusions with a general
closed convex state space. Preprint.}

\bibitem{RY}{Revuz D., Yor M. (1991). Continuous Martingales
and Brownian Motion. Springer-Verlag, Berlin.}

\bibitem{Sato}{Sato K. (1999). Levy Processes and Infinitely Divisible Distributions.
Cambridge University Press, Cambridge.}

\bibitem{SteinStein}{Stein E., Stein J. (1991). Stock price distributions with stochastic volatility:
an analytic approach. Review of Financial Studies 4, 727--752.}

\bibitem{Vasicek}{Vasicek O. (1977). An equilibrium characterisation of the term structure.
Journal of Financial Economics 5, 177–-188.}
\end{thebibliography}
\end{document}